\documentclass[10pt,conference]{IEEEtran}
\usepackage{amsfonts}
\usepackage{amsmath}
\usepackage{amssymb}
\usepackage{amsthm}
\usepackage{lscape}
\usepackage{epsf}
\usepackage{graphicx}
\usepackage{verbatim}
\usepackage{color} 

\newtheorem{theorem}{\indent Theorem}[section]

\newtheorem{proposition}[theorem]{\indent Proposition}

\newtheorem{EXAMPLE}{\indent Example}[section]
\newtheorem{definition}{\indent Definition}[section]

\newenvironment{example}{\begin{EXAMPLE}\rm}{\rm\end{EXAMPLE}}

\newcommand{\code}{{\mathcal{C}}}

\newcommand{\cB}{{\mathcal{B}}}
\newcommand{\cA}{{\mathcal{A}}}

\newcommand{\lsf}{{\mathsf{L}}}

\newcommand{\dist}{{\mathsf{d}}}

\newcommand{\sE}{{\mathsf{E}}}

\newcommand{\cay}{{\mathsf{Cay}}}

\newcommand{\Sn}{{\mathbb S}_n}

\newlength{\Algwidth}

\title{New Bounds for Permutation Codes in Ulam Metric}
\author{{\bf Faruk G\"olo\v{g}lu, J\"uri Lember, Ago-Erik Riet, and Vitaly Skachek} \\ 
Faculty of Mathematics and Computer Science\\ 
University of Tartu, Tartu 50409, Estonia}
\begin{document}

% \textbf{Keywords:} Pseudocodewords; Fundamental cone; Linear-programming decoder.

\maketitle

\begin{abstract}
New bounds on the cardinality of permutation codes equipped with the Ulam distance are presented. 
First, an integer-programming upper bound is derived, which improves on the Singleton-type upper bound 
in the literature for some lengths. Second, several probabilistic lower bounds are developed, which improve on the known lower 
bounds for large minimum distances. The results of a computer search for permutation codes are also presented. 
\end{abstract}


\section{Introduction}

A permutation code is a subset of the symmetric group $\Sn$, equipped with a distance metric. 
Permutation codes are of potential use in various applications, such as 
communications over Gaussian channels~\cite{32151}, \cite{1445610}, power-line communications~\cite{Blake19791}, \cite{1302307}, and coding for flash memories used with rank modulation~\cite{5205972}. 
Permutation codes were extensively studied in the literature over the last decades. In most of these studies, permutation codes are equipped with the Hamming and the Kendall $\tau$ metric~\cite{chadwick69}, \cite{diaconis1977spearman}, \cite{kendall1970rankcorrelation}.

Permutation codes were recently proposed for storing information in non-volatile (flash) memories under \emph{rank modulation}~\cite{5205972}, \cite{5452201}, \cite{5700265}, \cite{wadayama}. The main idea of the rank modulation scheme is that the information is stored in the form of rankings of the cell charges, rather than in terms of the absolute values of the charges. Rank-modulation codes represent a family of codes capable of handling errors of the form 
of \emph{adjacent transpositions}~\cite{5485013}, \cite{6034261} or \emph{translocations}~\cite{farnoud2013}. Such error patterns are typical for memory systems, where leakage of electric charge occurs over time.

There are two types of errors where a permutation code equipped with the Ulam or Kendall $\tau$ metric could be of use. One such type is overshoot errors, in which a cell receives more charge than it is supposed to. The second type is the errors, in which a defective cell loses charge more quickly than normal. Both of these types of errors constitute one error in the Ulam metric or a number of errors in the Kendall $\tau$ metric. Thus, codes in the Ulam or Kendall $\tau$ metrics seem to be appropriate for error detection and correction in the paradigm of rank modulation.

The problem of estimating the maximum size of a code in the Ulam metric for given parameters is very difficult. Different mathematical tools could be applied to this problem. In this work, we demonstrate that novel bounds on the maximum size of a code can be obtained by an integer-programming method and by probability estimation techniques. These two approaches deal with different regimes: the probability bounds are useful for large $n$ (some results, like Proposition~\ref{propLD}, are of the asymptotic nature only), whilst the integer-programming approach works well with (relatively) small $n$.

\section{Notation}

Denote by $\mathbb{Z}_{0}^{+}$ the set of non-negative integers. We also use the notation $[n] \triangleq \{ 1, 2, \cdots, n \}$. 

A \emph{permutation} $\sigma \, : \, [n] \rightarrow [n]$ is a bijection. Let $\Sn$ de\-no\-te the set of all permutations of the set $[n]$, i.e., the symmetric group of order $n!$. For any $\sigma \in \Sn$, we write $\sigma = [\sigma(1), \sigma(2), \cdots, \sigma(n)]$, where $\sigma(i)$ is the image of $i \in [n]$ under the permutation $\sigma$. This is called the one-line notation of permutation $\sigma$. The identity permutation $[1, 2, \cdots, n]$ is denoted by $e$, while $\sigma^{-1}$ stands for the inverse of the permutation~$\sigma$.  

Let $\dist\; : \; \Sn \times \Sn \rightarrow \mathbb{Z}_{0}^{+}$ be a metric defined for pairs of permutations. 
A \emph{permutation code} of length $n$ and minimum distance $d$ in a metric $\dist$ is a subset $\code$ of $\Sn$, such that for all $ \tau, \sigma \in \code$, $\tau \neq \sigma$, we have $\dist ( \tau, \sigma)\ge d$. Such a code will be also called an $(n, d)$ code in a metric $\dist$. 

\begin{definition}
Assume that $1\le i < j\le n$. A permutation $\tau \in \Sn$ is a \emph{right translocation} if
\[
\tau=[1, \cdots, i-1, i+1, i+2, \cdots, j, i, j+1, \cdots, n] \; . 
\]
A permutation $\tau \in \Sn$ is a \emph{left translocation} if 
\[ 
\tau = [1, \cdots, j-1, i, j, j+1, \cdots, i-1, i+1, \cdots, n] \; . 
\] 
\end{definition}

Next, we define the composition of two permutations. 
\begin{definition}
Let $\tau$ and $\sigma$ be two permutations in $\Sn$. Then, their composition $\tau \sigma$ is a permutation in $\Sn$ defined as 
\[
\forall i \in [n] \; : \; (\tau \sigma)(i) = \tau(\sigma(i)) \; . 
\]
\end{definition}

Under composition of permutations, $\Sn$ forms a group, called the symmetric group of order $n$.

\begin{definition} 

The \emph{Ulam distance} $\dist_U(\sigma, \rho)$ is the smallest integer $m$ such that there exists a sequence of (right and left) {\bf translocations} $\tau_1, \tau_2, ..., \tau_m$, such that $\rho = \sigma \tau_1 \tau_2 \cdots \tau_m$. 
\end{definition}

\begin{comment}
It turns out that the Ulam, Kendall $\tau$, cyclic Kendall $\tau$ and transposition metrics on $\Sn$ are all group metrics in the sense that $d(\sigma,\tau)$ is the graph distance of $\sigma^{-1}$ and $\tau^{-1}$ for different Cayley graphs. The generating sets consist of all left and right translocations for the Ulam metric, all adjacent transpositions for the Kendall $\tau$ metric, all adjacent transpositions and the maximum-distance transposition for the cyclic Kendall $\tau$ metric, and, all transpositions for the transposition metric. 

Another way to see that these metrics are graph metrics is to define the graph explicitly. Consider permutations in the one-line notation. For the Ulam metric, two vertices (permutations) are connected by an edge if one can be obtained from the other by picking up a symbol and inserting it somewhere else. The graphs for Kendall $\tau$, cyclic Kendall $\tau$ and transposition metrics are obtained by connecting two permutations by an edge if one can be obtained from the other by swapping symbols, respectively either adjacent symbols, adjacent symbols or the first and last symbols, or any two symbols. 

The Hamming metric is not even a graph metric. This is demonstrated by the fact that for the same $n$, the radius 2 balls around two distance 4 permutations may intersect for some pairs of permutations and not intersect for others, see for example \cite{SmMo}.
\end{comment} 

\begin{definition} 
A subsequence of length $m$ of $\sigma=[\sigma(1),\ldots,\sigma(n)]$ is a sequence of the form $\left[ \sigma(i_1),\ldots, \sigma(i_m)\right]$, where $i_1 < i_2 < \ldots i_m$. Let $\tau,\sigma \in \Sn$. The \emph{longest common subsequence} of $\tau$ and $\sigma$ is a subsequence of both 
$\tau$ and $\sigma$ of the longest possible length. 
\end{definition}

We denote the length of a longest common subsequence of $\tau$ and $\sigma$ by $\lsf \left(\tau,\sigma \right)$. Similarly, $\lsf (\sigma) \triangleq 
\lsf \left(\sigma, e \right)$, the length of a longest increasing subsequence of $\tau$. It is well known~\cite{farnoud2013} that for any $\sigma \in \Sn$, 
\begin{equation}
\dist_U ( \tau, \sigma ) = n - \lsf( \tau, \sigma) \; . 
\label{eq:lcs}
\end{equation}

\section{Integer-programming bound for the Ulam metric}

\subsection{Known bounds} 

Denote by $\cA(n,d)$ the maximum size of a code over $\Sn$  
equipped with the Ulam metric. The following theorem provides bounds
on $\cA(n,d)$~\cite{farnoud2013}. 

\begin{proposition}
\label{prop:tl-UB}
For all $n,d \in \mathbb{Z}_{0}^{+}$ with $n \ge d \ge 1$, 
\begin{equation}
\frac{\left(n-d+1\right)!}{\binom{n}{d-1}} \le \cA (n,d) \le (n-d+1) ! \; . 
\label{1}
\end{equation}
\end{proposition}

The right-hand side of~(\ref{1}) will be referred to as the \emph{Singleton bound} in the sequel.  
 
\subsection{Integer-programming bound}

In this section, we derive an integer-programming upper bound on $\cA (n,d)$. 

Let $\code \subseteq \Sn$ be a permutation code of Ulam distance $d$. It follows from~(\ref{eq:lcs}), that any subsequence 
of length $n-d+1$ appears at most once in any codeword of $\code$ (in other words, any two codewords in $\code$ cannot have the
same subsequence of length $n-d+1$ or more). We use this fact in order to define integer variables $X_{b,a}$ for all 
$1 \le a \le n$, $1 \le b \le n$. More specifically, 
\begin{eqnarray*}
   X_{b,a} = \left| \left\{ \sigma \in \code \; : \; \sigma(b) = a \right\} \right| \; . 
\end{eqnarray*}
In other words, $X_{b,a}$ counts a number of codewords with $a$ in position $b$. 

Assume that $\sigma \in \code$, such that $\sigma(b) = a$. Then, the number of different subsequences of $\sigma$ of length $n - d + 1$ of the form $(\underbrace{\bullet, \cdots, \bullet}_{\ell}, a, \underbrace{\bullet, \cdots, \bullet}_{n-d-\ell})$, where $\sigma(b) = a$ and $1 \le \ell \le n$, is given by 
\[
{b - 1 \choose \ell} \cdot { n - b \choose n - d - \ell} \; . 
\] 

On the other hand, there are $\frac{(n - 1)!}{(d-1)!}$ different sequences of the form $(\underbrace{\bullet, \cdots, \bullet}_{\ell}, a, \underbrace{\bullet, \cdots, \bullet}_{n-d-\ell})$. 

By a simple counting argument, we obtain that for all $a \in [n]$, 
\begin{equation}
\sum_{b = 1}^n {b - 1 \choose \ell} \cdot { n - b \choose n - d - \ell} \cdot X_{b,a} \le \frac{(n-1)!}{(d-1)!} \; .
\label{eq:constraint}
\end{equation}

The total number of the codewords can be obtained, for example, by $\sum_{a = 1}^n X_{b,a}$, for any $b \in [n]$. Therefore, we add constraints 
\[
\forall b \in[n-1] \; : \; \sum_{a = 1}^n X_{b,a} = \sum_{a = 1}^n X_{b+1,a} \; , 
\]
and an objective function
\[
\max \sum_{a = 1}^n X_{1,a} \; . 
\]

By combining this, we obtain the following linear program in Figure~\ref{fig:lp-general}, where its maximum provides an upper bound on $\cA (n,d)$. 

\begin{figure}[ht]
\hrule
\begin{eqnarray*}
\begin{array}{rll}
\mbox{\bf max} & \sum_{a = 1}^n X_{1,a} & \\
\mbox{\bf s.t.} & \forall a \in [n], \, \forall \ell \in [n-d+1] \; : & \\
&& \hspace{-21ex} \displaystyle \sum_{b = 1}^n {b - 1 \choose \ell} \cdot { n - b \choose n - d - \ell} \cdot X_{b,a} \le \frac{(n-1)!}{(d-1)!} \\
& \forall b \in[n-1] \; :  &  \displaystyle \sum_{a = 1}^n X_{b,a} = \sum_{a = 1}^n X_{b+1,a} \\ 
& \forall a,b \in [n] \; : & X_{b,a} \ge 0 
\end{array}
\end{eqnarray*}
\hrule
\caption{General integer-programming bound.}
\label{fig:lp-general}
\end{figure}

Next, observe that $X_{b,a}$ should be an integer. Therefore, we are interested in an integral solution to this linear-programming problem. 
This provides a tighter upper bound than the fractional solution to the same LP problem.  

\begin{example}
Take $n = 5$ and $d = 3$. The corresponding integer linear-programming problem is shown in Figure~\ref{fig:example}.  
\begin{figure}[ht]
\hrule
\begin{eqnarray*}
\begin{array}{rcl}
\mbox{\bf max} & \sum_{a=1}^5 X_{1,a} & \\
\mbox{\bf s.t.} & \forall a \in [5] \; : & \\
&& \hspace{-15ex} \displaystyle 1 \cdot {4 \choose 2} \cdot X_{1,a} + 1 \cdot {3 \choose 2} \cdot X_{2,a} + 1 \cdot {2 \choose 2} \cdot X_{3,a} \le 12 \\
                && \hspace{-5ex} \displaystyle 1 \cdot 3 \cdot X_{2,a} + 2 \cdot 2 \cdot X_{3,a} + 3 \cdot 1 \cdot X_{4,a} \le 12 \\
								&& \hspace{-15ex} \displaystyle {2 \choose 2} \cdot 1 \cdot X_{3,a} + {3 \choose 2} \cdot 1 \cdot X_{4,a} + {4 \choose 2} \cdot 1 \cdot X_{5,a} \le 12 \\
& \forall b \in [4] \; :  &  \displaystyle \sum_{a = 1}^5 X_{b,a} = \sum_{a = 1}^5 X_{b+1,a} \\ 
& \forall a,b \in [5] \; : & X_{b,a} \ge 0 
\end{array}
\end{eqnarray*}
\hrule
\caption{Integer program for $n=5$ and $d=3$.}
\label{fig:example}
\end{figure}

After simplification, this integer-programming problem becomes as in Figure~\ref{fig:simplified}. 

\begin{figure}[ht]
\hrule
\begin{eqnarray*}
\begin{array}{rcl}
\mbox{\bf max} & \sum_{a=1}^5 X_{1,a} & \\
\mbox{\bf s.t.} & \forall a \in [n] \; : & 6 X_{1,a} + 3 X_{2,a} + X_{3,a} \le 12 \\
                && 3 X_{2,a} + 4 X_{3,a} + 3 X_{4,a} \le 12 \\
								&& X_{3,a} + 3 X_{4,a} + 6 X_{5,a} \le 12 \\
& \forall b \in [4] \; :  &  \displaystyle \sum_{a = 1}^5 X_{b,a} = \sum_{a = 1}^5 X_{b+1,a} \\ 
& \forall a,b \in [5] \; : & X_{b,a} \ge 0 
\end{array}
\end{eqnarray*}
\hrule
\caption{Simplified integer program. }
\label{fig:simplified}
\end{figure}

By solving the integer-programming problem in Figure~\ref{fig:simplified}, we obtain that the maximum of the objective is obtained, for example, for $X_{b,a} = 1$ for all $a, b \in [n]$. 
This corresponds to the upper bound $\cA (n,d) \le 5$, which improves on the value $6$ obtained by using the Singleton bound. The actual 
value of $\cA (n,d)$ in this case is $4$. 
\end{example}

We remark, that the proposed integer LP problem can be further tightened by using additional constraints. For example, one can define additional 
variables $X_{(b_1, a_1), (b_2, a_2), \cdots, (b_t, a_t)}$, where all $a_i, b_i, t \in [n]$. Such a variable will count the number of permutations $\sigma$,
such that $\sigma(b_i) = a_i$ for all $i \in [t]$. Additional constraints can be defined in a manner similar to~(\ref{eq:constraint}), with respect to 
variables $X_{(b_1, a_1), (b_2, a_2), \cdots, (b_t, a_t)}$. 
%However, for small values of $n$ and $d$, the effect of these additional variables and constraints is negligible and it does not lead to tighter bounds. 

\section{Probabilistic bounds} 
\subsection{Asymptotic version of the lower bound}
In what follows, we consider an $(n,d)$ Ulam code. Denote $\Delta \triangleq d-1$. 
Recall the bounds in Proposition~\ref{prop:tl-UB}. 
By using
\begin{multline*}
m!\geq \left({m\over e}\right)^m=\exp\left[m(\ln m -1)\right] \\ \mbox{ and }
\quad {1\over m+1}\exp[m h_e(\alpha)] \leq \binom{m}{\alpha m} \leq
\exp[m h_e(\alpha)] \; ,
\end{multline*}
we obtain
\begin{equation}
{\left(n-\Delta\right)!\over \binom{n}{\Delta}} \; 
\geq \;  \exp \left[\left(n-\Delta\right)\left(\ln
(n-\Delta)-1\right)- n h_e\left({\Delta \over n}\right) \right] \; .
\label{tok}
\end{equation}
Here $h_e(p)$, where $p\in [0,1]$, is the binary entropy function with base $e$, i.e. 
$h_e(p) \triangleq -p\ln(p)-(1-p)\ln(1-p)$.
Hence,~(\ref{tok}) is an asymptotic lower bound on $\cA(n,d)$. 
Consider a special case of it, when $\Delta=n-c\sqrt{n}$, $c$ is a constant. Then $1-{\Delta\over n}={c/
\sqrt{n}}$, and (\ref{tok}) becomes
\begin{eqnarray}
&& \hspace{-5ex} \exp\left[\sqrt{n}c \left({1\over 2}\ln n+\ln c -1\right)-n h_e \left(1-{c\over
\sqrt{n}}\right)\right] \nonumber \\
& = & \exp\bigg[\sqrt{n}c \left(2\ln c -1\right) \nonumber \\
&& \hspace{8ex} + \; n\left(1-{c\over \sqrt{n}}\right) \ln \left(1-{c\over
\sqrt{n}}\right)\bigg] \nonumber \\
& \geq & \exp\left[\sqrt{n}c \left(2\ln c -1\right)-c\sqrt{n}\right] \nonumber \\
& = & \exp\left[2\sqrt{n}c \cdot \left(\ln c -1\right)\right] \; .
\label{eq:reference-bound}
\end{eqnarray}

Hence, with  $\Delta_n=n-c\sqrt{n}$, we have
\begin{equation}
\lim\inf_n {1\over \sqrt{n}}  \ln \left( {(c\sqrt{n})! \over \binom{n}{c\sqrt{n}}}\right) \; \geq \; 
2c \cdot (\ln c-1) \; .
\label{liminf}
\end{equation}

Let us now show that $2c \cdot (\ln c-1)$ is actually the limit. Indeed,
it holds
$m! \; = \; \left(1+o(1)\right)\sqrt{2\pi m}\left({m\over e}\right)^m$ (for large values of $m$).
Then, provided that $(n-\Delta)\to \infty$,
\begin{multline*}
\ln \left({(n-\Delta)! \over \binom{n}{\Delta}}\right) \; \leq \; 
(n-\Delta)\left(\ln (n-\Delta)-1\right) \\
+ \; {1\over 2}\ln
(2\pi(n-\Delta)) + \ln (1+o(1)) \\
- \; n h_e \left( {\Delta\over
n} \right)+\ln(n+1) \; .
\end{multline*} 
Hence, since $\Delta=n-c\sqrt{n}$,
\begin{align*}
&\ln \left({(c\sqrt{n})! \over \binom{n}{c\sqrt{n}}}\right) \; \leq \;\\
&  \hspace{4ex} 
c\sqrt{n} \left(2\ln c-1+ {\sqrt{n}\over c}(1-{c\over \sqrt{n}})\ln (1-{c\over
\sqrt{n}})\right) \\
& \hspace{8ex} \; + \; {1\over 2}\ln (2\pi c\sqrt{n})+\ln (1+o(1))+\ln({n}+1) \; .
\end{align*}
Since
$$
{\sqrt{n}\over c}\ln \left(1-{c\over \sqrt{n}}\right)\to -1 \; ,
$$
we obtain that
\begin{equation}
\limsup_n  {1\over \sqrt{n}} \ln \left({(c\sqrt{n})! \over
\binom{n}{c\sqrt{n}}}\right)\leq 2c(\ln c-1) \; .
\label{eq:lowerb}
\end{equation} 
By combining~(\ref{liminf}) with~(\ref{eq:lowerb}), we have the following result. 
\begin{proposition}
Let $\Delta=n-c\sqrt{n}$, where $c$ is a constant. Then, 
\begin{equation}
\lim_n {1\over \sqrt{n}}\ln \left({(n-\Delta_n)! \over \binom{n}{\Delta_n}}\right)
= 2c \cdot (\ln c-1) \; . 
\label{lim}
\end{equation}
\end{proposition}
%---------------

\subsection{Bounds using longest increasing subsequence} 

The following bounds hold for the Ulam metric (for the Kendall $\tau$ metric see, e.g.,~\cite{5485013}):
\begin{equation}\label{l1}
{n! \over |\cB(\Delta)|}\leq {\cal A}(n,d)\leq {n! \over  \left|\cB \left({\Delta \over 2} \right) \right|} \; ,
\end{equation}
where $|\cB(r)|\triangleq|\{\sigma : \dist_U(e,\sigma)\leq r\}|$ is the number of
permutations in the ball centered at the identity $e$ and having radius $r$. 
The
number of permutations in a ball $\cB(r)$ is difficult to estimate.
However, under the uniform distribution over all permutations in
$\mathbb{S}_n$, the ratio $|\cB(r)|/n!$ is just the probability that
a randomly chosen permutation is at distance at most $r$ from $e$. 
In terms of the longest increasing subsequences, thus,
\begin{equation}
\label{sphere}
{|\cB(r)|\over n!}=P(n-\lsf_n\leq r)=P(\lsf_n\geq n-r) \; ,
\end{equation} 
where $\lsf_n$ is the
length of a longest increasing subsequence of a random permutation under the 
uniform distribution. In terms of $\lsf_n$, the inequalities~(\ref{l1})
can be rewritten as
\begin{equation*}
{1\over P(\lsf_n\geq n-\Delta)}\leq {\cal A}(n,\Delta+1)\leq {1\over P(\lsf_n\geq n-\Delta/2)} \, .
%\label{l2}
\end{equation*}

By combining this with~(\ref{1}), when $\Delta$ is even, we obtain the following probability estimates
\begin{multline}
P(\lsf_n\geq n-\Delta/2)\leq {\binom{n}{\Delta} \over (n-\Delta)!} \\
\mbox{ and} \quad P(\lsf_n\geq n-\Delta)\geq {1 \over (n-\Delta)!}\;.
\label{l2a}
\end{multline}

The study of the properties of the random variable $\lsf_n$ has a
long history, starting with the pioneering paper of Ulam~\cite{ulam}, 
where the question of asymptotic behavior of $\sE[\lsf_n]$ was
stated. This so-called Ulam's problem deserved attention of many
researchers over several decades. In a sense, the problem was solved
by in the celebrated paper~\cite{BDJ}, 
where the limit law of (properly centered and scaled)
$\lsf_n$ was found. In particular, they showed that for every $t\in
\mathbb{R}$ (as $n$ increases),
\begin{equation}
\label{CLT}
P\left({\lsf_n-2\sqrt{n}\over n^{1\over 6}} \leq  t \right) \to F(t) \; ,
\end{equation}
where $F(t)$ is the distribution function of the Tracy-Widom law. For
a historical overview of Ulam's problem, the proof of~(\ref{CLT}), 
as well as the state of the art, we refer the reader to the book~\cite{lisbook}. 
Since the random variable $\lsf_n$ has been studied for
a relatively long time, one hopes that a proper upper estimate on
the probability $P(\lsf_n\geq n-\Delta)$ (or, alternatively, a lower estimate on the
probability $P(\lsf_n\geq n-\Delta/2)$) gives also a good lower (upper) bound
on ${\cal A}(n,d)$. 

\begin{comment}
The purpose of the current section is to show
that these probability estimates need to be really sharp, since most
of the well-known existing estimates give a worse lower bound than the
Singleton bound discussed above.
\end{comment}

%-------------
In what follows, we aim at bounding the probability $P(\lsf_n\geq n-\Delta)$ from
above. The following simple estimate can be found in \cite[page 9]{lisbook}:
\begin{equation}
\label{simple}
P(\lsf_n\geq n-\Delta)\leq  
{  \binom{n}{\Delta}\over
(n-\Delta)!} \; .
\end{equation} 
That estimate gives another proof of the lower bound~(\ref{1}). 
In order to improve it, the probability estimate has to be superior to~(\ref{simple}). 
When $n-2\geq \Delta\geq 1$, then the
inequality in (\ref{simple}) is  strict, and that follows from
the use of the Markov inequality in the proof. Hence, the lower bound in~(\ref{l1}) is always tighter than the bound~(\ref{1}). We have the following result. 
\begin{proposition}
The inequality
$$
{(n-\Delta)! \over \binom{n}{\Delta}} \leq {n! \over |\cB(\Delta)|} \; , 
$$
holds and for $0<\Delta<n-1$, the inequality is strict.\
\end{proposition}
\begin{proof} 
Apply (\ref{simple}) and (\ref{sphere}).
\end{proof}
\medskip

\subsubsection*{Bounds for $d = n - c \sqrt{n}$ } 
\begin{comment}
Weak bounds obtained via~(\ref{BDJ}) show that for our purposes much sharper
estimates on $P(\lsf_n>n-d)$ are needed. To the
best of our knowledge, such estimates exist only for $n-d=c\sqrt{n}$, where $c>2$ is a constant. 
In that case, the distance $d=n-c\sqrt{n}$ is rather large.
\end{comment}

One of the first probability estimates on $P(\lsf_n>n-d)$ was established
by Kim~\cite{kim}. Thus, for any $t \in \left(0, n^{1\over 3}/20 \right]$, it
holds that
\begin{equation}
\label{kim1}
P\left({\lsf_n-2\sqrt{n} } \geq t n^{1\over 6}\right)\leq  \exp \left[-{4\over
3} t^{3\over 2}+\phi(t) \right] \; ,
\end{equation}
where
\begin{equation*}
\phi(t)=\left({t\over 27
n^{1\over 3}}+{5\ln n\over t^{1\over 2}n^{1\over 3}}\right)t^{3\over
2} \; . 
\end{equation*} 
\begin{comment}
Therefore, when $t=n^{5\over 6}-2{n}^{1\over 3}-\Delta n^{-{1\over 6}}$, we obtain
\begin{multline*}
P(n-\lsf_n\leq \Delta) \; = \; P\left({\lsf_n-2\sqrt{n} } \geq t n^{1\over
6}\right) \\
\leq \exp\left[-(c-2)^{3\over 2}\sqrt{n}\left({38-c\over
27}-{5\ln n\over \sqrt{c-2}\sqrt{n}}\right)\right] \; ,
\end{multline*}
provided that $n-(2+{1\over 20})\sqrt{n} \leq \Delta < n-2\sqrt{n}$. 
\end{comment} 

That estimate leads to the lower bound on ${\cal A}(d,n)$ for $n - c \sqrt{n}$, where $c \in (2, 2 + 1/20]$, which is approximately
$$
\exp\left[ (c-2)^{3\over 2}\left({38-c \over 27}\right) \sqrt{n} \right] \; .
$$ 
This is the same order as the bound $\exp[2\sqrt{n}c (\ln c -1)]$, but the constant in the expression is smaller. 
This bound holds only for $c$ very close to 2 and above 2. 
%-----
\medskip

The best code rate estimate for large $d$ is given by the following large
deviation principle~\cite{BDJ}. For every $c>2$, 
\begin{equation}
\label{LD} 
\lim_n {1\over \sqrt{n}} \ln P \left(\lsf_n>c\sqrt{n} \right)  = -I(c) \; ,
\end{equation} 
where
\begin{eqnarray}
-I(c) & = & -2c \,\,{\rm cosh}^{-1}\left({c\over 2}\right)+2\sqrt{c^2-4} \nonumber \\ 
& = & -2c\ln \left({c\over 2}+\sqrt{{c^2\over 4}-1}\right)+2\sqrt{c^2-4} \; . \hspace{5ex} 
\label{i}
\end{eqnarray} 
In terms of the lower bound, (\ref{LD}) can be stated as follows.

\begin{proposition}\label{propLD} For every constant $c>2$, the following convergence
holds: 
\begin{eqnarray*}
\lim_n {1\over \sqrt{n}} \ln \left({n!\over
|\cB(\Delta_n)|}\right) & = & I(c)>2c(\ln c-1) \\
& = & \lim_n {1\over \sqrt{n}} \ln
\left({(n-\Delta_n)! \over \binom{n}{\Delta_n}}\right) \; ,
\end{eqnarray*}
where $\Delta_n=n-c\sqrt{n}-1$ and $-I(c)$ is given in (\ref{i}).
\end{proposition}
%----------
\begin{proof} Use (\ref{sphere}) together with (\ref{LD}) and
(\ref{lim}). Note that (\ref{lim}) is formally proven for
$\Delta_n=n-c\sqrt{n}$, but it also holds  for
$\Delta_n=n-c\sqrt{n}-1$.
\end{proof}
This proposition yields an asymptotic improvement on the lower bound in~(\ref{1}).

\begin{comment}
Proposition~\ref{propLD} states that for large $n$,
\begin{multline*}
{n!\over |\cB(\Delta_n)|} \; \approx \; 
\exp[\sqrt{n}I(c)] \\
\; > \; \exp \left[ \sqrt{n}2c(\ln c-1) \right ] \; \approx \; {(n-\Delta_n)!
\over \binom{n}{\Delta_n}} \; .
\end{multline*} 
\end{comment}

We note that any probability
estimate that is better than the very simple estimate in~(\ref{simple}),
gives a better lower bound on ${\cal A}(n,d)$ in comparison with the
existing lower bound~(\ref{1}). 
Except for large $d$, there are no better estimates known. 
On the other hand, any good upper bound on $|\cB(n-\Delta)|$
entails also a good estimate on the probability $P(\lsf_n\geq \Delta)$.
Since the probabilities $P(\lsf_n\geq \Delta)$ are closely related to the Tracy-Widom
distribution, such a link between coding and probability theory might be valuable.

\section{Computational results} 

In this section, we present computational results related to the optimal codes. It turns out that there exist non-trivial Ulam-metric codes, which attain the Singleton bound with equality. We call such codes \emph{Singleton-optimal}. Singleton-optimal Ulam $(n,d)$ codes are also known as \emph{perfect deletion-correcting codes on $n$ distinct symbols, capable of correcting $d-1$ deletions} and as \emph{directed Steiner systems}~\cite{Levenshtein}. They exist for every $n$ and $d=2$~\cite{Levenshtein}, and also for $n=6$ and $d=3$~\cite{Mathon}. It was also found by the exhaustive search in~\cite{Mathon} that $\cA (7,4) = 12$. We have complemented these results for other pairs $(n,d)$.

Table~\ref{tab:size} summarizes what is known about $\cA (n,d)$. These results are obtained by using computer search, and they improve on the theoretical bounds in many cases. Table~\ref{tab:yes} summarizes the experimental results on the existence of the Singleton-optimal $(n,d)$ Ulam codes.

\begin{table*}[t]
\setlength{\tabcolsep}{12pt}
\def\arraystretch{1.3}
\begin{center}
\begin{tabular}{ r | c c c c c c } 
\hline \hline
          & $d = 2$ & $d = 3$ & $d = 4$ & $d = 5$ & $d = 6$ & $d = 7$ \\ \hline 
$n = 4$ &          6 & 2 & -- & -- & -- & -- \\ 
$n = 5$ &         24 & 4 & 2 & -- & -- & -- \\ 
$n = 6$ &         120 & 24 & 4 & 2 & -- & -- \\ 
$n = 7$ &         720 & $\geq 59 \mbox{ and} < 120 $  & $12$  & 4 & 2 & -- \\  
$n = 8$ &         5040 & ? & $<120$ & $\leq 12$ & 4 & 2 \\ 
$n = 9$ &         40320 & ? & ? & $<120$ & $\leq 12$ & 2 \\ \hline \hline
\end{tabular}        
\end{center}
\caption{Known maximum sizes of codes in the Ulam metric.}
\label{tab:size}
\end{table*}

{\small
\setlength{\tabcolsep}{8pt}
\def\arraystretch{1.3}
\begin{table}[hb]
\begin{center}
\begin{tabular}{ r | c  c  c  c  c  c } 
\hline \hline
       & $d = 2$ & $d = 3$ & $d = 4$ & $d = 5$ & $d = 6$ & $d = 7$ \\ 
			\hline 
$ n = 4$ &  yes & yes & -- & -- & -- & -- \\ 
$ n = 5$ &          yes & no & yes & -- & -- & -- \\ 
$ n = 6$ &           yes & yes & no & yes & -- & -- \\ 
$ n = 7$ &           yes & no & no & no & yes & -- \\ 
$ n = 8$ &         yes & ? & no & no & no & yes \\ 
$ n = 9$ &         yes & ? & ? & no & no & no \\ 
\hline \hline
\end{tabular}		
\end{center}
\caption{The existence of Singleton-optimal codes in the Ulam metric.}
\label{tab:yes}
\end{table}
}

In order to obtain these results, we construct the graph on the vertex set $\Sn$ with an edge if and only if the corresponding vertices are at least a distance $d$ away. Our goal is to find a clique of the maximum size.

We assign colors to the vertices of this graph, such that color class of a permutation corresponds to the relative ordering of symbols $1,2,...,n-d+1$ in the one-line notation of the permutation. The existence of a Singleton-optimal code becomes equivalent to the property that the clique number of the graph is equal to its chromatic number. Thus, in order to obtain a Singleton-optimal code, we need to pick exactly one vertex from each color class, such that the induced graph forms a clique. 

To obtain a maximum-size code when a Singleton-optimal code does not exist, we need to pick at most one vertex from each color class. This makes the respective exhaustive search computationally much harder.

\vspace{-1ex}

\section{Acknowledgements}

The work of the authors is supported in part by the Estonian Research Council through the research grants PUT405, PUT620, IUT2-1, IUT20-57 and IUT34-5, by the Estonian Science Foundation through the grant ETF9288, and by the European Regional Development Fund through the Estonian Center of Excellence in Computer Science, EXCS. The authors wish to thank Dirk Oliver Theis for helpful discussions.

%------------------------------------------------------------------------------------------------

%---------------------------------------------------------------------------------------------

%====================================================================================================================

\begin{thebibliography} {99}

\bibitem{BDJ}
J. Baik, P. Deift, and K. Johanson,
\newblock {``On the distribution of the length of the longest increasing subsequences of random permutations,''}
\newblock {\em Journal of the American Math. Society}, vol. 12, no. 4, pp. 1119 - 1178, 1999.

\bibitem{5485013}
A.~Barg and A.~Mazumdar, ``Codes in permutations and error correction for rank
  modulation,'' \emph{IEEE Trans. on Inform. Theory}, vol.~56, no.~7, pp. 3158
  --3165, Jul. 2010.

\bibitem{Blake19791}
I.F. Blake, G.~Cohen, and M.~Deza, ``Coding with permutations,''
  \emph{Information and Control}, vol.~43, no.~1, pp. 1--19, 1979.

\bibitem{5205972}
J.~Bruck, A.~Jiang, and Z.~Wang, ``On the capacity of bounded rank modulation
  for flash memories,'' in \emph{Proc. IEEE Intern. Symp. on Inform. Theory}, Jun./Jul. 2009, pp. 1234 --1238.

\bibitem{chadwick69}
H.~Chadwick and L.~Kurz, ``Rank permutation group codes based on Kendall's
  correlation statistic,'' \emph{IEEE Trans. on Inform. Theory}, vol.~15,
  no.~2, pp. 306 -- 315, Mar. 1969.

\bibitem{Klove}
J.C. Chang, R.J. Chen, T.~Klove, and S.C. Tsai, ``On the maximum number of
  permutations with given maximal or minimal distance,'' 
	\emph{IEEE Trans. on Inform. Theory}, vol. 49(4), pp. 1054--1059, 2003.

\bibitem{1302307}
C.J. Colbourn, T.~Klove, and A.C.H. Ling, ``Permutation arrays for powerline
  communication and mutually orthogonal latin squares,'' 
	\emph{IEEE Trans. on Inform. Theory}, vol.~50, no.~6, pp. 1289 -- 1291, Jun. 2004.

\bibitem{diaconis1977spearman}
P.~Diaconis and R.~Graham, ``Spearman's footrule as a measure of disarray,''
  \emph{J. Roy. Statistical Soc. Series B}, vol.~39, no.~2,
  pp. 262--268, 1977.

\bibitem{farnoud2013}
F.~Farnoud, V.~Skachek, and O.~Milenkovic, ``Error-correction in flash memories via codes in the Ulam metric,'' 
\emph{IEEE Trans. on Inform. Theory}, vol. 59, no. 5, pp. 3003-3020, May 2013.

\bibitem{5452201}
A.~Jiang, M.~Schwartz, and J.~Bruck, ``Correcting charge-constrained errors in
  the rank-modulation scheme,'' \emph{IEEE Trans. on Inform. Theory}, 
	vol.~56, no.~5, pp. 2112 --2120, May 2010.

\bibitem{5700265}
A.~Jiang and Y.~Wang, ``Rank modulation with multiplicity,'' in \emph{Proc. IEEE Globecom
  Workshops}, Dec. 2010, pp. 1866 --1870.

\bibitem{32151}
J.~Karlof, ``Permutation codes for the Gaussian channel,'' 
\emph{IEEE Trans. on Inform. Theory}, vol.~35, no.~4, pp. 726 --732, Jul. 1989.

\bibitem{kendall1970rankcorrelation}
M.~Kendall, \emph{Rank correlation methods}, 4th~ed. London: Griffin, 1970.

\bibitem{kim}
J.H. Kim,
\newblock {``On increasing subsequences of random permutations,''}
\newblock {\em Journal of Combinatorial Theory}, vol. 76, pp. 148 - 155, 1996.

\bibitem{Levenshtein} 
V.I. Levenshtein, 
\newblock {``On perfect codes in deletion and insertion metric,''}
\newblock {\em Discrete Math. Appl.}, vol. 2, no. 3, pp. 241 - 258, 1992.

\bibitem{LiHa} 
F. Lim and M. Hagiwara, ``Linear programming upper bounds on permutation code sizes from coherent configurations related to the Kendall-tau distance metric,'' \emph{Proc. IEEE Intern. Symp. on Inform. Theory,} 2012.

\bibitem{Mathon}
R. Mathon and T. van Trung,
\newblock {``Directed $t$-packings and directed $t$-Steiner systems,''}
\emph{Designs, Codes and Cryptography}, vol.~18, no.~1-3, pp. 187--198, 1999.

\bibitem{6034261}
A.~Mazumdar, A.~Barg, and G.~Zemor, ``Constructions of rank modulation codes,''
  in \emph{Proc. IEEE Intern. Symp. on Inform. Theory}, 
	Jul./Aug. 2011, pp. 869 --873.

\bibitem{lisbook}
D. Romik,
\newblock {\em The Surprising Mathematics of Longest Increasing Subsequence},
\newblock {Cambridge}, 2014.

\bibitem{1445610}
D.~Slepian, ``Permutation modulation,'' \emph{Proceedings of the IEEE},
  vol.~53, no.~3, pp. 228 -- 236, Mar. 1965.

\bibitem{SmMo}
D. H. Smith and R. Montemanni, ``Permutation codes with specified packing radius,'' \emph{Designs, Codes and Cryptography}, vol.~69, no.~1, pp. 95 --106, 2013.

\bibitem{ulam}
S. Ulam,
\newblock {``Monte-Carlo calculations in problems of mathematical physics,''}
\newblock {\em Modern Mathematics for the Engineer, Second Series}, (E. Beckenbach, ed.), pp. 261 - 281, 1961.

\bibitem{wadayama}
T. Wadayama and M. Hagiwara, 
\newblock {``LP-decodable permutation codes based on linearly constrained permutation matrices,''}
\emph{IEEE Trans. on Inform. Theory}, vol. 58, no. 8, pp. 5454--5470, Aug. 2012.

\bibitem{bruck2010partial-rank-modulation}
Z.~Wang and J.~Bruck, ``Partial rank modulation for flash memories,'' in
  \emph{Proc. IEEE Intern. Symp. on Inform. Theory}, Jun. 2010, pp. 864 --868.

\end{thebibliography}
\end{document}